\documentclass[3p,11pt]{elsarticle}

\usepackage{graphicx,epsfig}
\usepackage{url}
\usepackage{multirow,array}
\usepackage{fixltx2e,amsmath,graphicx,amssymb,breakurl,amsthm,latexsym}

\newcommand{\Suff}{\textit{Suff}}
\newcommand{\Pref}{\textit{Pref}}
\newcommand{\Fact}{\textit{Fact}}
\renewcommand{\alph}{\textit{alph}}
\newcommand{\St}{\textit{St}}
\newcommand{\LS}{\textit{LS}}
\newcommand{\RS}{\textit{RS}}
\newcommand{\BS}{\textit{BS}}
\newcommand{\SBS}{\textit{SBS}}
\newcommand{\WBS}{\textit{NBS}}
\newcommand{\MF}{\textit{MF}_{St}}
\newcommand{\M}{\textit{MF}_{L}}
\renewcommand{\epsilon}{\varepsilon}

\newtheorem{theorem}{Theorem}[section]
\newtheorem{proposition}[theorem]{Proposition}
\newtheorem{lemma}[theorem]{Lemma}
\newtheorem{corollary}[theorem]{Corollary}
\theoremstyle{definition}
\newtheorem{definition}{Definition}
\newtheorem{example}{Example}

\newtheorem{remark}{Remark}

\begin{document}

\begin{frontmatter}

\title{ \textbf{ On the Structure of Bispecial Sturmian Words }\tnoteref{note1} }

\tnotetext[note1]{A preliminary version of this paper was presented at the \emph{37th
  International Symposium on Mathematical Foundations of Computer Science, MFCS 2012} \cite{Fi12}.}

\author{Gabriele Fici}
\ead{gabriele.fici@unipa.it}
\address{Dipartimento di Matematica e Informatica, Universit\`a di Palermo\\Via Archirafi 34, 90123 Palermo, Italy}

\journal{Journal of Computer and System Sciences}

\begin{abstract}
A balanced word is one in which any two factors of the same length contain the same number of each letter of the alphabet up to one. Finite binary balanced words are called Sturmian words. A Sturmian word is bispecial if it can be extended to the left and to the right with both letters remaining a Sturmian word. There is a deep relation between bispecial Sturmian words and Christoffel words, that are the digital approximations of Euclidean segments in the plane. In 1997, J. Berstel and A. de Luca proved that \emph{palindromic} bispecial Sturmian words are precisely the maximal internal factors of \emph{primitive} Christoffel words. We extend this result by showing that bispecial Sturmian words are precisely the maximal internal factors of \emph{all} Christoffel words. Our characterization allows us to give an enumerative formula for bispecial Sturmian words. We also investigate the minimal forbidden words for the language of Sturmian words.
\end{abstract}


\end{frontmatter}

\section{Introduction}\label{sec:intro}

A word $w$ is balanced if and only if for any $u,v$ factors of $w$ of the same length, and for any letter $a$, one has $||u|_{a}-|v|_{a}|\le 1$, where $|z|_{a}$ denotes the number of $a$'s in the word $z$. 

Balanced words appear in several problems in Computer Science. For example, Altman, Gaujal and  Hordijk \cite{AlGaHo00} proved that balanced words are optimal sequences for some classes of scheduling problems, such as routing among several systems. An interesting problem arising in this context is that of constructing infinite balanced words with assigned frequencies of letters. There is a conjecture of A. S. Fraenkel  \cite{Fr73}, originally stated in the context of Number Theory, that is equivalent to the following: for any fixed $k > 2$, there is only one infinite balanced  word (up to letter permutation) over an alphabet of size $k$, in which all letters have different frequencies, and this word is periodic. The Fraenkel Conjecture has been proved true for small alphabet sizes (see \cite{Vu03,BaVa03} and references therein), but the general problem remains open. 

For any alphabet $\Sigma$ of size at least two, there exist infinite words over $\Sigma$ that are balanced and aperiodic. When $|\Sigma|=2$, these are called infinite Sturmian words. Sturmian words are very rich from the combinatorial point of view, and because of this fact they have a lot of equivalent definitions and characterizations (see, as a classical reference, \cite[Chapter~2]{LothaireAlg}).
However, if the Fraenkel Conjecture is true for every $k>2$, the only balanced infinite words that are aperiodic and have different letter frequencies are the infinite Sturmian words.

A finite Sturmian word (or, briefly, a Sturmian word) is a factor of some infinite Sturmian word. The set $\St$ of Sturmian words therefore coincides with the set of binary balanced finite words.

If one considers extendibility within the set $\St$ of Sturmian words, one can define left special Sturmian words (resp.~right special Sturmian words) \cite{DelMi94} as those words $w$ over the alphabet $\Sigma=\{a,b\}$ such that $aw$ and $bw$ (resp.~$wa$ and $wb$) are both Sturmian words. For example, the word $aab$ is left special since $aaab$ and $baab$ are both Sturmian words, but is not right special since $aabb$ is not a Sturmian word.

Left special Sturmian words are precisely the binary words having suffix automaton\footnote{The suffix automaton of a finite word $w$ is the minimal deterministic finite state automaton accepting the set of suffixes of $w$.} with minimal state complexity (cf.~\cite{SciZa07,Fi10b}). From combinatorial considerations one has that right special Sturmian words are the reversals of left special Sturmian words.

The Sturmian words that are both left and right special are called bispecial Sturmian words. They are of two kinds: strictly bispecial Sturmian words, that are the words $w$ such that $awa$, $awb$, $bwa$ and $bwb$ are all Sturmian words (e.g.\ $aa$), or non-strictly bispecial Sturmian words otherwise (e.g.\ $ab$). Strictly bispecial Sturmian words are also called central words, and have been deeply studied (see for example~\cite{DelMi94,CarDel05}) because they constitute the kernel of the theory of Sturmian words. Non-strictly bispecial Sturmian words, instead, received less attention.

One important field in which Sturmian words arise naturally is Discrete Geometry. Indeed, infinite Sturmian words can be viewed as the digital approximations of Euclidean straight lines in the plane. It is known that given a point $(p,q)$ in the grid $\mathbb{Z} \times \mathbb{Z}$, with $p,q>0$, there exists a unique path that approximates from below (resp.~from above) the Euclidean segment joining the origin $(0,0)$ to the point $(p,q)$. If one encodes horizontal and vertical unitary segments with the letters $a$ and $b$ respectively, this path is called the lower (resp.~upper) Christoffel word\footnote{Some authors require that $p$ and $q$ be coprime in the definition of Christoffel word. Here we follow the definition given in \cite{BeDel97} and do not require this condition.} associated to the pair $(p,q)$, and is denoted by $w_{p,q}$ (resp.~$w'_{p,q}$). By elementary geometrical considerations, one has that for any $p,q>0$, $w_{p,q}=aub$ for some word $u$, and $w'_{p,q}=b\tilde{u}a$, where $\tilde{u}$ is the reversal of  $u$. If (and only if) $p$ and $q$ are coprime, the Christoffel words $w_{p,q}$ and $w'_{p,q}$ are primitive (that is, they are not a concatenation of copies of a shorter word). 

A well known result of Jean Berstel and Aldo de Luca \cite{BeDel97} is that a word $u$ is a strictly bispecial Sturmian word if and only if $aub$ is a primitive lower Christoffel word (or, equivalently, if and only if $bua$ is a primitive upper Christoffel word). As a main result of this paper, we show that this correspondence holds in general between bispecial Sturmian words and Christoffel words. More precisely, we prove (in Theorem \ref{theor:main}) that $u$ is a bispecial Sturmian word if and only if there exist letters $x,y$ in $\{a,b\}$ such that $xuy$ is a Christoffel word. 

This characterization allows us to prove an enumerative formula for bispecial Sturmian words (Corollary \ref{cor:formula}): there are exactly $2n+2-\phi(n+2)$ bispecial Sturmian words of length $n$, where $\phi$ is the Euler totient function, i.e., $\phi(n)$ is the number of positive integers smaller than or equal to $n$ and coprime with $n$. Surprisingly, enumerative formulae for left special, right special and strictly bispecial Sturmian words were known \cite{DelMi94}, but to the best of our knowledge we exhibit the first proof of an enumerative formula for non-strictly bispecial (and therefore for bispecial) Sturmian words.

We then investigate the minimal forbidden words for the set of finite Sturmian words. Recall that the set of minimal forbidden words of a factorial language is the set of words of minimal length that do not belong to the language \cite{MiReSci02}. More precisely, given a factorial language $L$ over an alphabet $\Sigma$, a word $v=v_{1}v_{2}\cdots v_{n}$, with $v_{i}\in \Sigma$, is a minimal forbidden word for $L$ if $v_{1}\cdots v_{n-1}$ and $v_{2}\cdots v_{n}$ are in $L$, but $v$ is not.

Minimal forbidden words represent a powerful tool to investigate the structure of a factorial language (see~\cite{BeMiRe96,BeMiReSc00,CrMiRe98}), such as the language of factors of a (finite or infinite) word, or of a set of words. They also appear in different contexts in Computer Science, such as symbolic dynamics \cite{BeMiReSc00}, data compression (where the set of minimal forbidden words is often called an antidictionary) \cite{CrMiReSa99}, or bio-informatics (where they are also called minimal absent words) \cite{ChCr12}.

We give a characterization of minimal forbidden words for the language $\St$ of Sturmian words in Theorem \ref{theor:mf}. We show that they are precisely the words of the form $ywx$ such that $xwy$ is a non-primitive Christoffel word, where $\{x,y\}=\{a,b\}$. This characterization allows us to give an enumerative formula for the set of minimal forbidden words of $\St$ (Corollary \ref{cor:formulamf}):  there are exactly $2(n-1-\phi(n))$ minimal forbidden words of length $n$ for every $n>1$.

The paper is organized as follows. In Section~\ref{sec:wsf} we recall standard definitions on words and factors. In Section~\ref{sec:StCh} we deal with Sturmian words and Christoffel words, and present our main result, and in Section~\ref{sec:En} we give an enumerative formula for bispecial Sturmian words. Finally, in Section~\ref{sec:MF}, we investigate minimal forbidden words for the language of finite Sturmian words. 

\section{Words and special factors}\label{sec:wsf}

We give here basic definitions on words and fix the notation.

An alphabet, denoted by $\Sigma$, is a finite set of symbols, called letters. A word over $\Sigma$ is a finite sequence of letters from $\Sigma$. The length of a word $w$ is denoted by $|w|$. The only word of length $0$ is called the empty word and is denoted by $\epsilon$.  The set of all words over $\Sigma$ is denoted by $\Sigma^*$. The set of all words over $\Sigma$ having length $n$ is denoted by $\Sigma^n$. Any subset $X$ of $\Sigma^{*}$ is called a language, and we note $X(n)=|X\cap \Sigma^{n}|$ the set of words of length $n$ in $X$. 

Given a non-empty word $w$, we denote its $i$-th letter by $w[i]$, $1 \le i \le |w|$. The reversal of the word $w=w[1]w[2]\cdots w[n]$ is the word $\tilde{w}=w[n]w[n-1]\cdots w[1]$. We set $\tilde{\epsilon}=\epsilon$. A palindrome is a word $w$ such that $\tilde{w}=w$. A word is called a power if it is the concatenation of copies of another word; otherwise it is called primitive. For a letter $a\in \Sigma$, $|w|_{a}$ is the number of $a$'s occurring in $w$. A positive integer $p$ is a period of a word $w$ if $p>|w|$ or $w[i]=w[i+p]$ for every $i=1,\ldots ,|w|-p$. For a word $au$, $a\in \Sigma$, $u\in \Sigma^{*}$, we define $\rho(au)=ua$. The set of rotations of a word $w$ of length $n$ is the set $\{\rho^{i}(w) \mid 1\le i \le n\}$. Note that the rotations of a word $w$ are all different if and only if $w$ is primitive.

A word $z$ is a factor of a word $w$ if $w=uzv$ for some $u,v\in \Sigma^{*}$. In the special case $u = \varepsilon $ (resp.~$v = \varepsilon $), we call $z$ a prefix (resp.~a suffix) of $w$. We let $\Pref(w)$, $\Suff(w)$ and $\Fact(w)$ denote the set of prefixes, suffixes and factors of the word $w$, respectively. The factor complexity of a word $w$ is the integer function $f_{w}(n)=|\Fact(w)\cap \Sigma^n|$,  $n\geq 0$. 

A factor $u$ of a (finite or infinite) word $w$ is called left special (resp.~right special) in $w$ if there exist $a,b\in \Sigma$, $a\neq b$, such that $au,bu\in \Fact(w)$ (resp.~$ua,ub\in \Fact(w)$). A bispecial factor is a factor that is both left and right special. Moreover, a bispecial factor $u$ of a word $w$ is strictly bispecial if $xuy$ is a factor of $w$ for every $x,y\in \Sigma$; otherwise $u$ is non-strictly bispecial. For example, let $w= aababba$. The left special factors of $w$ are $\epsilon$, $a$, $ab$, $b$ and $ba$. The right special factors of $w$ are $\epsilon$, $a$, $ab$ and $b$. Therefore, the bispecial factors of $w$ are $\epsilon$, $a$, $ab$ and $b$. Among these, only $\epsilon$ is strictly bispecial.

\section{Sturmian words and Christoffel words}\label{sec:StCh}

In the rest of the paper, unless otherwise specified, we fix the alphabet $\Sigma=\{a,b\}$.

A word $w\in \Sigma^{*}$ is called Sturmian if it is balanced, i.e., if for any $u,v$ factors of $w$ of the same length, one has $||u|_{a}-|v|_{a}|\le 1$ (or, equivalently, $||u|_{b}-|v|_{b}|\le 1$). We let $\St$ denote the set of Sturmian words. The language $\St$ is factorial (i.e., if $w=uv\in \St$, then $u,v\in \St$) and extendible (i.e., for every $w\in \St$, there exist letters $x,y\in \Sigma$ such that $xwy\in \St$). 

The following definitions are in \cite{DelMi94}. 

\begin{definition}
A Sturmian word $w\in \St$ is  left special (resp.~right special) if and only if $aw$, $bw\in \St$ (resp.~if $wa$, $wb\in \St$). A bispecial Sturmian word is a Sturmian word that is both left and right special. Moreover, a bispecial Sturmian word is  strictly bispecial if and only if $awa$, $awb$, $bwa$ and $bwb$ are all Sturmian word; otherwise it is non-strictly bispecial. 
\end{definition}

\begin{remark}
 The definition of special Sturmian word is different from the (widely studied) definition of special \emph{factor} of an infinite Sturmian word (see \cite[Definition 10]{DelMi94}). Actually, a word is a bispecial factor of some infinite Sturmian word if and only if it is a strictly bispecial Sturmian word. 
\end{remark}

We let $\LS$, $\RS$, $\BS$, $\SBS$ and $\WBS$ denote the sets of left special, right special, bispecial, strictly bispecial and non-strictly bispecial Sturmian words, respectively. Thus, one has $\BS=\LS\cap \RS=\SBS\cup \WBS$.

The following lemma is a reformulation of a result of Aldo de Luca \cite{Del97}.

\begin{lemma}\label{lem:prefsuf}
Let $w$ be a word over $\Sigma$. Then $w\in \LS$ (resp.~$w\in \RS$) if and only if $w$ is a prefix (resp.~a suffix) of a word in $\SBS$.
\end{lemma}

Given a bispecial Sturmian word, the simplest criterion to determine if it is strictly or non-strictly bispecial  is provided by the following nice characterization \cite{DelMi94}:

\begin{proposition}\label{prop:sturmstrispe}
 A bispecial Sturmian word is strictly bispecial if and only if it is a palindrome.
\end{proposition}

Another useful result is the following (\cite[Lemma 7]{DelMi94}).

\begin{lemma}\label{lem:delmi}
 If $awb$ and $bwa$ are both in $\St$, then $awa$ and $bwb$ also are, i.e., $w$ is strictly bispecial.
\end{lemma}

We can now derive the following classification of Sturmian words with respect to their extendibility.

\begin{proposition}\label{prop:bisp}
 Let $w$ be a Sturmian word. Then:
 
\begin{enumerate}
\item  $|\Sigma w\Sigma\cap \St|=4$ if and only if $w$ is strictly bispecial;
\item  $|\Sigma w\Sigma\cap \St|=3$ if and only if $w$ is non-strictly bispecial;
\item  $|\Sigma w\Sigma\cap \St|=2$ if and only if $w$ is left special or right special but not bispecial;
\item  $|\Sigma w\Sigma\cap \St|=1$ if and only if $w$ is neither left special nor right special.
\end{enumerate}
\end{proposition}

\begin{proof}
 1. and 4. follow from the definitions. 
 
 For 3., if $awb$ and $bwa$ are both in $\St$, then by Lemma \ref{lem:delmi}, $awa$ and $bwb$ also are, and then in this case $|\Sigma w\Sigma\cap \St|=4$.  On the other hand, it is also known that if $awa$ and $bwb$ are in $\St$, then at least one between $awb$ and $bwa$ is in $\St$, as a consequence of the fact that if $w$ is right special, then there exists a letter $x$ such that $xw$ is right special (see \cite[Lemma 8]{DelMi94}), and then in this case $|\Sigma w\Sigma\cap \St|\ge 3$. So, the only possible cases for $|\Sigma w\Sigma\cap \St|=2$ are when $\Sigma w\Sigma\cap \St=\{xwx,xwy\}$ or $\Sigma w\Sigma\cap \St=\{xwx,ywx\}$, for different letters $x$ and $y$, i.e., when  $w$ is left special or right special but not bispecial. 
 
  For 2., always by Lemma \ref{lem:delmi}, we have that the only possible case for $|\Sigma w\Sigma\cap \St|=3$ is when $awa$, $bwb$ and only one between $awb$ and $bwa$ are in $\St$, i.e., when $w$ is bispecial but not strictly.
\end{proof}

Note that, in general, if $L$ is a factorial and extendible language over $\Sigma$ and $w\in L$, one can have that $w$ is bispecial in $L$ and $|\Sigma w\Sigma\cap L|=2$ (for example if $\Sigma w\Sigma=\{awa,bwb\}$ or if $\Sigma w\Sigma=\{awb,bwa\}$). The previous proposition shows  that this cannot happen for the language $\St$.  

We now recall the definition of central word \cite{DelMi94}.

\begin{definition}
A word over $\Sigma$ is central if it has periods $p$ and $q$, with $\gcd(p,q)=1$, and length equal to $p+q-2$.
\end{definition}

\begin{example}
The word $w=aba$ is central, since it has periods $2$ and $3$ and length $3$. The empty word is central (in this case $p=q=1$); in fact, any word of the form $w=a^{n}$, $a\in \Sigma$, $n\ge 0$, is central, since it has periods $p=1$ and $q=n+1$ and length $p+q-2$.
\end{example}

A useful combinatorial characterization of central words is the following (see~\cite{Del97}):

\begin{proposition}\label{prop:charcentral}
A word $w$ over $\Sigma$ is central if and only if $w$ is the power of a single letter or there exist palindromes $P,Q$ such that $w=PxyQ=QyxP$, for different letters $x,y\in \Sigma$. Moreover, if $|P|<|Q|$, then $Q$ is the longest palindromic suffix of $w$.
\end{proposition}

Actually, in the statement of Proposition \ref{prop:charcentral}, the requirement that the words $P$ and $Q$  are palindromes is not even necessary \cite{CarDel05}.

We have the following remarkable result \cite{DelMi94}:

\begin{proposition}\label{prop:sbscen}
A word over $\Sigma$ is a strictly bispecial Sturmian word if and only if it is a central word.
\end{proposition}

We now introduce Christoffel words, that are words coding the digital approximations of segments in the Euclidean plane.

\begin{definition}
Let $n>1$ and $p,q>0$ be integers such that $p+q=n$. The lower Christoffel word $w_{p,q}$ is the word defined for $1\le i\le n$ by
\[w_{p,q}[i] = \left\{ \begin{array}{lllll}
a & \mbox{if $iq$ mod$(n)>(i-1)q$ mod$(n)$,}\\
b & \mbox{if $iq$ mod$(n)<(i-1)q$ mod$(n)$.}\\
\end{array} \right.\]
\end{definition}

\begin{example}
 Let $p=6$ and $q=4$. We have $\{i4$ mod$(10)\mid i=0,1,\ldots,10\}=\{0,4,8,2,6,0,4,8,2,6,0\}$. Hence, $w_{6,4}=aababaabab$.
\end{example}

Notice that for every $n>1$, there are exactly $n-1$ lower Christoffel words $w_{p,q}$, corresponding to the $n-1$ pairs $(p,q)$ such that $p,q>0$ and $p+q=n$. 

\begin{remark}
In the literature, Christoffel words are often defined with the additional requirement that $\gcd(p,q)=1$ (cf.~\cite{Book08}). We call such Christoffel words primitive, since a Christoffel word is a primitive word if and only if $\gcd(p,q)=1$.
\end{remark}

If one draws a word in the grid $\mathbb{Z} \times \mathbb{Z}$ by encoding each $a$ with a horizontal unitary segment and each $b$ with a vertical unitary segment, the lower Christoffel word $w_{p,q}$ is in fact the best grid approximation from below of the Euclidean segment joining $(0,0)$ to $(p,q)$, and has slope $q/p$, that is, $|w|_{a}=p$ and  $|w|_{b}=q$ (see~\figurename~\ref{fig:GC}).

\begin{figure}
\begin{center}
\begin{minipage}{7.2cm}
\includegraphics[height=50mm]{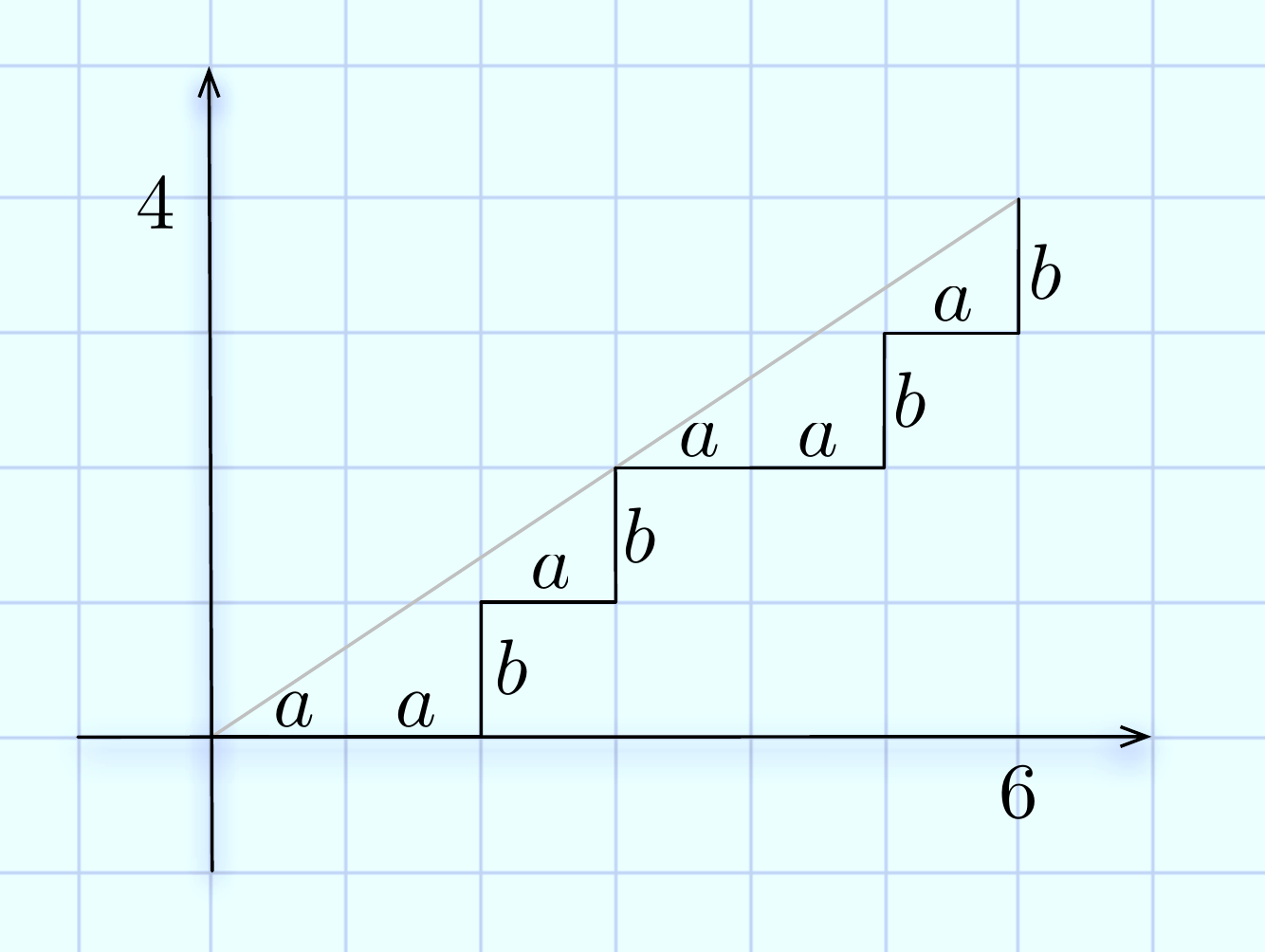}
\end{minipage}
\begin{minipage}{7.2cm}
\includegraphics[height=50mm]{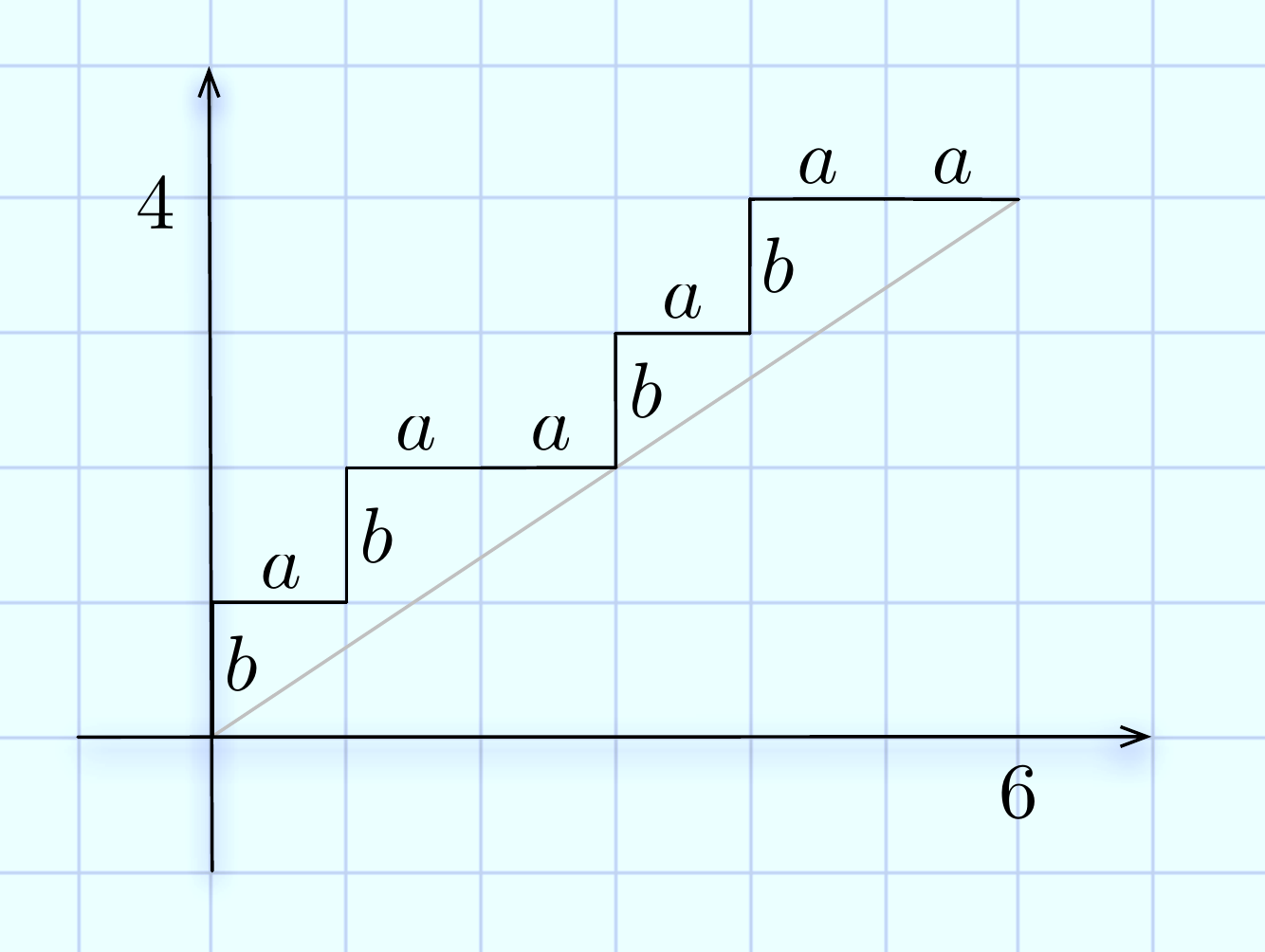}
\end{minipage}
\end{center}
\caption{The lower Christoffel word $w_{6,4}=aababaabab$ (left) and the upper Christoffel word $w'_{6,4}=babaababaa$ (right).\label{fig:GC}}
\end{figure}

Analogously, one can define the upper Christoffel word $w'_{p,q}$ by
\[w'_{p,q}[i] = \left\{ \begin{array}{lllll}
a & \mbox{if $ip$ mod$(n)<(i-1)p$ mod$(n)$,}\\
b & \mbox{if $ip$ mod$(n)>(i-1)p$ mod$(n)$.}\\
\end{array} \right.\]
Of course, the upper Christoffel word $w'_{p,q}$ is the best grid approximation from above of the Euclidean segment joining $(0,0)$ to $(p,q)$ (see \figurename~\ref{fig:GC}).

\begin{example}
 Let $p=6$ and $q=4$. We have $\{i6$ mod$(10)\mid i=0,1,\ldots,10\}=\{0,6,2,8,4,0,6,2,8,4,0\}$. Hence, $w'_{6,4}=babaababaa$.
\end{example}

The next result follows from elementary geometrical considerations.

\begin{lemma}\label{lem:rev}
 For every pair of positive integers $(p,q)$ the upper Christoffel word $w'_{p,q}$ is the reversal of the lower Christoffel word $w_{p,q}$.
\end{lemma}

If (and only if) $p$ and $q$ are coprime, the Christoffel word $w_{p,q}$ (resp.~$w'_{p,q}$) intersects the Euclidean segment joining $(0,0)$ to $(p,q)$ only at the end points, and is a primitive word. 

The link between primitive Christoffel words and central words (i.e., by Proposition \ref{prop:sbscen}, strictly bispecial Sturmian words) is contained in the following remarkable result of Jean Berstel and Aldo de Luca (cf.~\cite{BeDel97}).

\begin{theorem}\label{theor:sbsCP}
  $\SBS=\{w\mid xwy \mbox{ is a primitive Christoffel word, $x,y\in \Sigma\}$}.$
\end{theorem}

If instead the integers $p$ and $q$ are not coprime, then there exist coprime integers $p',q'$ such that $p=rp'$, $q=rq'$, for an integer $r>1$. In this case, we have $w_{p,q}=(w_{p',q'})^{r}$, i.e., $w_{p,q}$ is a power of a primitive Christoffel word. Hence, by Theorem \ref{theor:sbsCP}, there exists a central word $u$ such that $w_{p,q}=(aub)^{r}$ and $w'_{p,q}=(bua)^{r}$. So, we have:

\begin{lemma}\label{lem:npChris}
The word $xwy$, $x\neq y\in \Sigma$, is a Christoffel word if and only if $w=(uyx)^{n}u$, for an integer $n\ge 0$ and a central word $u$. Moreover, $xwy$ is a primitive Christoffel word if and only if $n=0$.
\end{lemma}

Another way to construct Christoffel words is the following. Let $p,q$ be positive integers and  $p+q=n$. Write a block of $p$ $a$'s followed by a block of $q$ $b$'s. Now shift the block of $b$'s by $q$ positions on the left, modulo $n$, and fit the remaining positions with the $p$ $a$'s. Repeating this procedure $n$ times, one obtains $n$ words of length $n$, that can be arranged to form the rows of a square matrix $A_{p,q}$ (see \figurename~\ref{fig:matrix} for an example). The first column of the matrix $A_{p,q}$ is the lower Christoffel word $w_{p,q}$, while the last column is the upper Christoffel word $w'_{p,q}$. Actually, the columns of $A_{p,q}$ are precisely the $n$ rotations of $w_{p,q}$ and are lexicographically ordered from left to right. Note that the $n$ columns are all distinct if and only if the Christoffel word $w_{p,q}$ is primitive, i.e., if and only if $p$ and $q$ are coprime. This construction is linked to the fact that Christoffel words have a completely clustered Burros-Wheeler transform (see \cite{MaReSc03} for more details). 

\renewcommand\arraystretch{.9}

\begin{figure}
\begin{center}
\newcolumntype{C}{>{\centering\arraybackslash}b{.1mm}<{}}
$$
A_{6,4}=
\left(
{\begin{tabular}{*{10}{C}}
a & a & a & a & a & a & b & b & b & b \\
a & a & b & b & b & b & a & a & a & a \\
b & b & a & a & a & a & a & a & b & b \\
a & a & a & a & b & b & b & b & a & a \\
b & b & b & b & a & a & a & a & a & a \\
a & a & a & a & a & a & b & b & b & b \\
a & a & b & b & b & b & a & a & a & a \\
b & b & a & a & a & a & a & a & b & b \\
a & a & a & a & b & b & b & b & a & a \\
b & b & b & b & a & a & a & a & a & a \\
\end{tabular}}
\hspace{1.5mm}\right) $$
\end{center}
\caption{The matrix $A_{p,q}$ for $p=6$ and $q=4$. The first column is the lower Christoffel word $w_{6,4}$, the last column is the upper Christoffel word $w'_{6,4}$. The columns of $A_{p,q}$ are precisely the rotations of the word $w_{6,4}$, and appear lexicographically ordered from left to right.\label{fig:matrix}}
\end{figure}

\bigskip

We are now extending the result in Theorem \ref{theor:sbsCP} to non-primitive Christoffel words.

Recall from \cite{Del97} that the right (resp.~left) palindromic closure of a word $w$ is the (unique) shortest palindrome $w^{(+)}$ (resp.~$w^{(-)}$) such that $w$ is a prefix of $w^{(+)}$ (resp.~a suffix of $w^{(-)}$). If $w=uv$ and $v$ is the longest palindromic suffix of $w$ (resp.~$u$ is the longest palindromic prefix of $w$), then $w^{(+)}=w\tilde{u}$ (resp.~$w^{(-)}=\tilde{v}w$). 

\begin{lemma}\label{lem:rpl}
 Let $xwy$ be a Christoffel word, $x,y\in \Sigma$. Then  $w^{(+)}$ and  $w^{(-)}$ are central words.
\end{lemma}

\begin{proof}
 Let $xwy$ be a Christoffel word, $x,y\in \Sigma$.  By Lemma \ref{lem:npChris}, $w=(uyx)^{n}u$, for an integer $n\ge 0$ and a central word $u$. We prove the statement for the right palindromic closure, that for the left palindromic closure will then follow by symmetry. If $n=0$, then $w=u$, so $w$ is a palindrome and then $w^{(+)}=w$ is a central word. So suppose $n>0$. We first consider the case when $u$ is the power of a single letter (including the case $u=\epsilon$). We have that either $w=(y^{k+1}x)^{n}y^{k}$ or $w=(x^{k}yx)^{n}x^{k}$ for some $k\ge 0$. In the first case, $w^{(+)}=wy=(y^{k+1}x)^{n}y^{k+1}$, whereas in the second case $w^{(+)}=wyx^{k}=(x^{k}yx)^{n}x^{k}yx^{k}$. In both cases one has that $w^{(+)}$ is a strictly bispecial Sturmian word, and thus, by Proposition \ref{prop:sbscen}, a central word.

Let now $u$ be not the power of a single letter. Hence, by Proposition \ref{prop:charcentral}, there exist palindromes $P,Q$ such that $u=PxyQ=QyxP$. Now, observe that 
\[
 w  =  (uyx)^{n}u
   =  Pxy(QyxPxy)^{n}Q.
\]
We claim that the longest palindromic suffix of $w$ is $(QyxPxy)^{n}Q$. Indeed, the longest palindromic suffix of $w$ cannot be $w$ itself since $w$ is not a palindrome, so since any  palindromic suffix of $w$ longer than $(QyxPxy)^{n}Q$ must start in $u$, in order to prove the claim it is enough to show that the first non-prefix occurrence of $u$ in $w$ is that appearing as a prefix of $(QyxPxy)^{n}Q$. Now, since the prefix $v=PxyQyxP$ of $w$ can be written as $v=uyxP=Pxyu$, one has by Proposition \ref{prop:charcentral} that $v$ is a central word. It is easy to prove (see, for example, \cite{BuDelFi12}) that the longest palindromic suffix of a central word does not have internal occurrences, i.e., appears in the central word only as a prefix and as a suffix. Therefore, since $|u|>|P|$, $u$ is the longest palindromic suffix of $v$ (by Proposition \ref{prop:charcentral}), and so appears in $v$ only as a prefix and as a suffix. This shows that $(QyxPxy)^{n}Q$ is the longest palindromic suffix of $w$.

Thus, we have $w^{(+)}=wyxP$, and we can write: 
\begin{eqnarray*}
 w^{(+)} & = & Pxy(QyxPxy)^{n}QyxP\\
 & = & PxyQ \cdot yx \cdot P(xyQyxP)^{n}\\
 & = & (PxyQyx)^{n}P \cdot xy \cdot QyxP,
\end{eqnarray*}
so that $w^{(+)}=uyxz=zxyu$ for the palindrome $z=P(xyQyxP)^{n}=(PxyQyx)^{n}P$. By Proposition \ref{prop:charcentral}, $w^{(+)}$ is a central word.
\end{proof}

We are now ready to state our main result.

\begin{theorem}\label{theor:main}
$\BS=\{w\mid xwy \mbox{ is a Christoffel word, $x,y\in \Sigma\}$}.$
\end{theorem}

\begin{proof}
 Let $xwy$ be a Christoffel word, $x,y\in \Sigma$.  Then, by Lemma \ref{lem:npChris}, $w$ is of the form $w=(uyx)^{n}u$, $n\ge 0$, for a central word $u$. By Lemma \ref{lem:rpl}, $w$ is a prefix of the central word $w^{(+)}$ and a suffix of the central word $w^{(-)}$, and therefore, by Proposition \ref{prop:sbscen} and Lemma \ref{lem:prefsuf}, $w$ is a bispecial Sturmian word.
 
 Conversely, let $w$ be a bispecial Sturmian word, that is, suppose that the words $xw$, $yw$, $wx$ and $wy$ are all Sturmian. If $w$ is strictly bispecial, then $w$ is a central word by Proposition \ref{prop:sbscen}, and $xwy$ is a (primitive) Christoffel word by Theorem \ref{theor:sbsCP}. So suppose $w\in \WBS$. By Lemma \ref{lem:npChris}, it is enough to prove that $w$ is of the form $w=(uyx)^{n}u$, $n\ge 1$, for a central word $u$ and letters $x\neq y$. Since $w$ is not a strictly bispecial Sturmian word, it is not a palindrome (by Proposition \ref{prop:sturmstrispe}). Let $u$ be the longest palindromic prefix of $w$ that is also a suffix of $w$, so that $w=uyzxu$, $x\neq y\in \Sigma$, $z\in \Sigma^{*}$. If $z=\varepsilon$, $w=uyxu$ and we are done. Otherwise, it must be $z=xz'y$ for some $z'\in \Sigma^{*}$, since otherwise either the word $yw$ would contain $yuy$ and $xxu$ as factors (a contradiction with the hypothesis that $yw$ is a Sturmian word) or the word $wx$ would contain $uyy$ and $xux$ as factors (a contradiction with the hypothesis that $wx$ is a Sturmian word). So $w=uyxz'yxu$. If $u=\varepsilon$, then it must be $z=(yx)^{k}$ for some $k\ge 0$, since otherwise either  $xx$ would appear as a factor in $w$, and therefore the word $yw$ would contain $xx$ and $yy$ as factors, being not a Sturmian word, or $yy$ would appear as a factor in $w$, and therefore the word $wx$ would contain $xx$ and $yy$ as factors, being not a Sturmian word. Hence, if $u=\epsilon$ we are done, and so we now suppose $|u|>0$. 
 
 By contradiction, suppose that $w$ is not of the form $w=(uyx)^{n}u$. That is, let $w=(uyx)^{k}u'av$, with $k\ge 1$, $v\in \Sigma^{*}$, $u'b\in \Pref(uyx)$, for different letters $a$ and $b$. If $|u'|\ge |u|$, then either $|u'|=|u|$ or $|u'|=|u|+1$. In the first case, $u'=u$ and $w=(uyx)^{k}uxv'$, for some $v'\in \Sigma^{*}$, and then the word $yw$ would contain $yuy$ and $xux$ as factors, being not a Sturmian word. In the second case, $u'=uy$ and $w=(uyx)^{k}uyyv''$, for some $v''\in \Sigma^{*}$; since $xu$ is a suffix of $w$, and therefore $w=(uyx)^{k}v'''xu$ for some $v'''\in \Sigma^{*}$, we would have that the word $wx$ contains both $uyy$ and $xux$ as factors, being not a Sturmian word. Thus, we can suppose $u'b\in \Pref(u)$. Now, if $a=x$ and $b=y$, then the word $yw$ would contain the factors $yu'y$ and $xu'x$, being not a Sturmian word; if instead $a=y$ and $b=x$, let $u=u'xu''$, so that we can write $w=(uyx)^{k}u'yv=(uyx)^{k-1}u'xu''yxu'yv$. The word $wx$ would therefore contain the factors $u''yxu'y$ and $xux=xu'xu''x$ (since $xu$ is a suffix of $w$), being not a Sturmian word (see~\figurename~\ref{fig:theorem}). In all the cases we obtain a contradiction and the proof is thus complete.
\end{proof}

\begin{figure}[ht]
\begin{center}
\includegraphics[height=42mm]{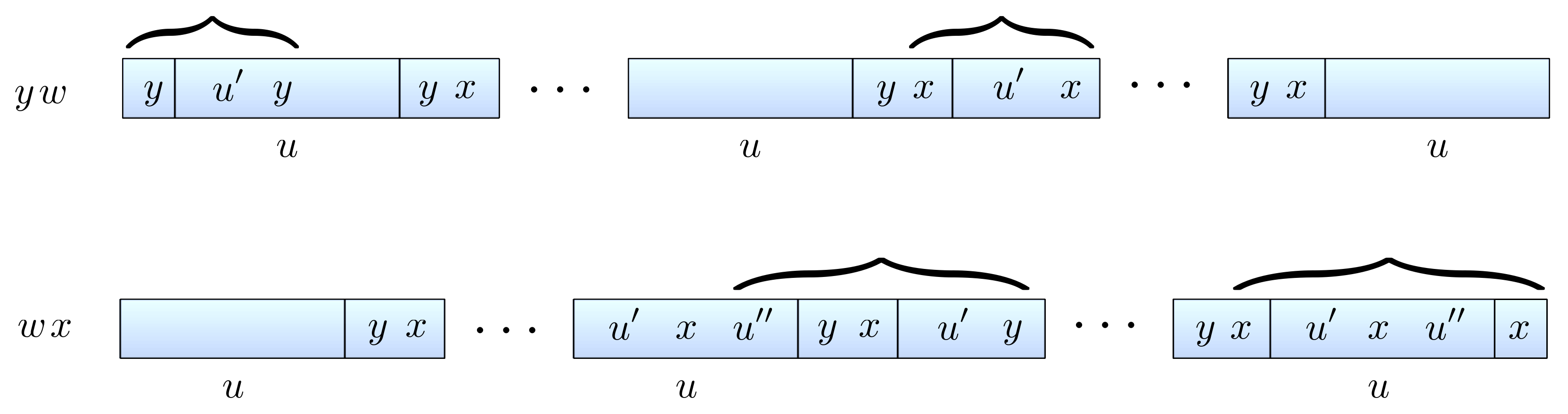}
\caption{An illustration of the proof of Theorem \ref{theor:main}.}
\label{fig:theorem}
\end{center}
\end{figure}

So, bispecial Sturmian words are the maximal internal factors of Christoffel words. Every bispecial Sturmian word is therefore of the form $w=(uyx)^{n}u$, $n\ge 0$, for different letters $x,y$ and a central word $u$. The word $w$ is strictly bispecial if and only if $n=0$. If $n=1$, $w$ is a semicentral word \cite{BuDelFi12}, i.e., a word in which the longest repeated prefix, the longest repeated suffix, the longest left special factor and the longest right special factor all coincide.

\section{Enumeration of bispecial Sturmian words}\label{sec:En}

In this section we give an enumerative formula for bispecial Sturmian words. It is known that the number of Sturmian words of length $n$ is given by
\begin{equation}\label{eq:St}
 St(n)=1+\sum_{i=1}^{n}(n-i+1)\phi(i),
\end{equation}
where $\phi$ is the Euler totient function, i.e., $\phi(n)$ is the number of positive integers smaller than or equal to $n$ and coprime with $n$ (cf.~\cite{Mig91,Lip82}).

Let $w$ be a Sturmian word of length $n$. If $w$ is left special, then $aw$ and $bw$ are Sturmian words of length $n+1$. If instead $w$ if not left special, then only one between $aw$ and $bw$ is a Sturmian word of length $n+1$. Therefore, we have $\LS(n)=St(n+1)-St(n),$ and hence
\begin{equation}
 \LS(n)=\sum_{i=1}^{n+1}\phi(i).
\end{equation}
Using a symmetric argument, one has that also 
\begin{equation}
 \RS(n)=\sum_{i=1}^{n+1}\phi(i).
\end{equation}
Since \cite{DelMi94} $\SBS(n)=\LS(n+1)-\LS(n)=\RS(n+1)-\RS(n)$, we obtain
\begin{equation}
 \SBS(n)=\phi(n+2).
\end{equation}
Therefore, in order to find an enumerative formula for bispecial Sturmian words, we only have to enumerate the non-strictly bispecial Sturmian words. We do it in the next proposition.

\begin{proposition}
For every $n>1$, one has 
\begin{equation}
 \WBS(n)=2(n+1-\phi(n+2)).
\end{equation}
\end{proposition}

\begin{proof}
Let \[W_{n}=\{w \mid \mbox{  $awb$ is a lower Christoffel word of length $n+2$} \},\] and \[W'_{n}=\{w' \mid \mbox{  $bw'a$ is an upper Christoffel word of length $n+2$} \}.\] By Theorem \ref{theor:main}, the bispecial Sturmian words of length $n$ are the words in $W_{n}\cup W'_{n}$. 

Among the $n+1$ words in $W_{n}$, there are $\phi(n+2)$ strictly bispecial Sturmian words, that are precisely the palindromes in $W_{n}$. The $n+1-\phi(n+2)$ words in $W_{n}$ that are not palindromes are non-strictly bispecial Sturmian words. The other non-strictly bispecial Sturmian words of length $n$ are the $n+1-\phi(n+2)$ words in $W'_{n}$ that are not palindromes. Since the words in $W'_{n}$ are the reversals of the words in $W_{n}$, and since no non-strictly bispecial Sturmian word is a palindrome by Proposition \ref{prop:sturmstrispe}, there are a total of $2(n+1-\phi(n+2))$ non-strictly bispecial Sturmian words of length $n$.
\end{proof}

\begin{corollary}\label{cor:formula}
 For every $n\ge 0$, there are $2(n+1)-\phi(n+2)$ bispecial Sturmian words of length $n$.
\end{corollary}

\begin{example}
The Christoffel words of length $12$ and their maximal internal factors, the bispecial Sturmian words of length $10$, are reported in Table \ref{tab:example} (the strictly bispecial Sturmian words are underlined). 

\begin{table}[h]
\begin{center}
  \begin{tabular}{|c | c | c | }
  
\ Pair $(p,q)$ \ & \ Lower Christoffel word $w_{p,q}$ \ & \ Upper Christoffel word $w'_{p,q}$ \  \\    \hline 
 $(11,1)$ &    $a\underline{aaaaaaaaaa}b$   & $b\underline{aaaaaaaaaa}a$     \\
 $(10,2)$ &    $aaaaabaaaaab$   & $baaaaabaaaaa$       \\
 $(9,3)$  &    $aaabaaabaaab$   & $baaabaaabaaa$       \\
 $(8,4)$  &    $aabaabaabaab$   & $baabaabaabaa$       \\
 $(7,5)$  &    $a\underline{ababaababa}b$  & $b\underline{ababaababa}a$       \\
 $(6,6)$  &    $abababababab$  & $babababababa$   \\
 $(5,7)$  &    $a\underline{bababbabab}b$   & $b\underline{bababbabab}a$       \\
 $(4,8)$  &    $abbabbabbabb$   & $bbabbabbabba$       \\
 $(3,9)$  &    $abbbabbbabbb$   & $bbbabbbabbba$       \\
 $(2,10)$ &    $abbbbbabbbbb$   & $bbbbbabbbbba$       \\
 $(1,11)$ &    $a\underline{bbbbbbbbbb}b$   & $b\underline{bbbbbbbbbb}a$      \\    
 \hline 
  \end{tabular}\vspace{4mm}
\end{center}\caption{The Christoffel words of length $12$. Their maximal internal factors are the bispecial Sturmian words of length $10$. There are $4=\phi(12)$ strictly bispecial Sturmian words, that are the palindromes $aaaaaaaaaa$, $ababaababa$, $bababbabab$ and $bbbbbbbbbb$ (underlined), and $14=2(11-4)$ non-strictly bispecial Sturmian words: $aaaaabaaaa$, $aaaabaaaaa$, $aaabaaabaa$, $aabaaabaaa$, $aabaabaaba$, $abaabaabaa$, $ababababab$, $bababababa$, $babbabbabb$, $bbabbabbab$, $bbabbbabbb$, $bbbabbbabb$, $bbbbabbbbb$ and  $bbbbbabbbb$.\label{tab:example}}
\end{table}
\end{example}

\section{Minimal forbidden words}\label{sec:MF}

Recall that a language $L$ is called factorial if it contains all the factors of its words (i.e., for every words $u,v$ such that $uv\in L$ one has that $u$ and $v$ belong to $L$), and extendible if every of its words has arbitrarily long extensions on the left and on the right (i.e., for every $u\in L$, there exist letters $x$ and $y$ such that $xuy\in L$). 

Given a factorial language $L$ over an alphabet $\Sigma$, a word $w\in \Sigma^{*}$ is a minimal forbidden word for $L$ if $w$ does not belong to $L$ but every proper factor of $w$ does. The set $\M$ of minimal forbidden words for $L$ is therefore defined by the equation:
\begin{equation}
 \M=\Sigma L \cap L \Sigma \cap (\Sigma^{*}\setminus L).
\end{equation}

Note that given a word $w$ in a factorial and extendible language $L$ over an alphabet $\Sigma$, if $w$ is a maximal internal factor of a minimal forbidden word for $L$, then $w$ is a bispecial word in $L$. In fact, let $x,y\in \Sigma$ be such that $xwy$ is a minimal forbidden word for $L$. By the definition of minimal forbidden word, $xw$ and $wy$ belong to $L$. Since $L$ is extendible, there is a letter $y'\neq y$ in $\Sigma$ such that $xwy'\in L$; symmetrically, there is a letter $x'\neq x\in \Sigma$ such that $x'wy\in L$. Since $L$ is factorial, $wy$, $wy'$, $xw$ and $x'w$ belong to $L$, and therefore $w$ is bispecial in $L$. 
 However, the converse is not true: if $w$ is a strictly bispecial word in $L$, then it cannot be the maximal internal factor of a minimal forbidden word for $L$.

In the next theorem, we give a characterization of the set $\MF$ of minimal forbidden words for the language $\St$ of finite Sturmian words. 

\begin{theorem}\label{theor:mf}
 $\MF=\{ywx \mid xwy \mbox{ is a non-primitive Christoffel word, $x,y\in \Sigma$}\}.$
\end{theorem}

\begin{proof}
If $xwy$ is a non-primitive Christoffel word, then by Theorems \ref{theor:sbsCP} and  \ref{theor:main}, $w$ is a non-strictly bispecial Sturmian word. This implies that $ywx$ is not a Sturmian word, otherwise by Lemma \ref{lem:delmi} we have a contradiction. Since $yw$ and $wx$ are Sturmian words, we have $ywx\in \MF$.

Conversely, let $ywx\in \MF$. We claim that $x$ and $y$ are different letters. In fact, suppose that $ywx=awa$ (the case $ywx=bwb$ is symmetric). Then, by the definition of minimal forbidden word, $aw$ and $wa$ are Sturmian words, and since $\St$ is extendible, $awb$ and $bwa$ are Sturmian. By Lemma \ref{lem:delmi}, $w$ is strictly bispecial, and thus $awa$ is Sturmian, a contradiction. So we can suppose that $x\neq y$. Therefore, we have that $ywy$ and $xwx$ are Sturmian words. Thus, $w$ is a bispecial Sturmian word, and since $ywx\notin \St$, $w$ is a  non-strictly bispecial Sturmian word. By Theorems \ref{theor:sbsCP} and \ref{theor:main}, $xwy$ is a non-primitive Christoffel word.
\end{proof}

\begin{corollary}\label{cor:formulamf}
 For every $n>1$, one has 
\begin{equation}\label{eq:MF}
  \MF(n)=2(n-1-\phi(n)).
\end{equation}
\end{corollary}

It is known from \cite{Mig91} that $\St(n)=O(n^{3})$, as a consequence of \eqref{eq:St} and of the estimation (see~\cite{HaWr}, p. 268)
\begin{equation}\label{eq:HW}
 \sum_{i=1}^{n}\phi(i)=\frac{3n^{2}}{\pi^{2}}+O(n\log n).
\end{equation}
So, we have that the number of Sturmian words of length smaller than or equal to $n$ is $\sum_{i=1}^{n}\St(i)=O(n^{4})$.
From \eqref{eq:HW} and \eqref{eq:MF}, we have that the number of minimal forbidden words of length smaller than or equal to $n$ is $\sum_{i=1}^{n}\MF(n)=O(n^{2})$.

\section{Conclusions}

We studied the combinatorics of bispecial Sturmian words, and their relations with Christoffel words and digital approximations of segments in the plane. A natural question is that of extending this setting to higher dimensions, that is, to alphabets of cardinality greater than $2$. There exist several generalizations of the definition of infinite Sturmian word to larger alphabets, each of which captures only a part of the properties of these words. 

The most studied generalization is the notion of episturmian word. An infinite word over an alphabet of cardinality greater than 2 is episturmian if it has at most one left special factor for each length and the set of its factors is closed under reversal (see \cite{GlJu09} for a survey on episturmian words). Unfortunately, episturmian words are not necessarily balanced, and even balanced ones (G.~Paquin and L.~Vuillon \cite{PaVu07} gave a combinatorial characterization of balanced episturmian words) do not seem to correspond precisely to the approximations of the straight lines---see the introduction of \cite{BeLa11} for more references.

Another generalization of Sturmian words in higher dimensions are billiard words. Billiard words  have the property that they are represented by a path that lies at bounded distance from a Euclidean straight line. Billiard words are moreover balanced, but in dimension higher than $2$ they do not have linear complexity of factors  \cite{Ar94}---which is a fundamental property of Sturmian words. 

A generalization of the notion of balancedness is $C$-balancedness. Given an integer $C>0$, a word is said to be $C$-balanced if for any pair of factors of the same length $u$ and $v$, and for any letter $a$, one has $||u|_{a}-|v|_{a}|\leq C$. So what we referred to in this paper as balanced words are precisely the $1$-balanced words.
V. Berth\'e and S. Labb\'e raised the question whether it is possible to construct $C$-balanced words with linear factor complexity that have prescribed letter frequencies \cite{BeLa11a}. The same authors also presented interesting results about the algorithmic generation of digital approximations of segments in the 3-dimensional space \cite{BeLa11}. 

We believe that an approach based on the Combinatorics on Words can lead to further insights on problems of digital approximation of lines in the 3-dimensional space.

\section*{Acknowledgements}

The author warmly thanks  Alessandro De Luca, S\'ebastien Labb\'e, Filippo Mignosi and Antonio Restivo for fruitful discussions. This work has been partially supported by PRIN 2010/2011 project ``Automi e Linguaggi Formali: Aspetti Matematici e Applicativi'' of the Italian Ministry of Education (MIUR).

\end{document}